\documentclass[reqno]{amsart}
\usepackage{amssymb}
\usepackage{hyperref}
\usepackage{booktabs}
\usepackage{multirow}
\usepackage{tikz}

\newtheorem{theorem}{Theorem}[section]
\newtheorem{proposition}[theorem]{Proposition}

\theoremstyle{remark}
\newtheorem{remark}[theorem]{Remark}

\numberwithin{equation}{section}

\begin{document}

\title[Spectrum of an open $q$-difference Toda chain]
{Spectral analysis of an open $q$-difference Toda chain with two-sided boundary interactions on the finite integer lattice}

\author{Jan Felipe  van Diejen}

\address{
Instituto de Matem\'aticas, Universidad de Talca,
Casilla 747, Talca, Chile}

\email{diejen@inst-mat.utalca.cl}


\date{January 2023}

\begin{abstract} 
A quantum $n$-particle model consisting of an open $q$-difference Toda chain  with two-sided boundary interactions is placed on a finite integer lattice.
The spectrum and eigenbasis are computed by establishing the equivalence with a previously studied $q$-boson model from which the quantum integrability is inherited. Specifically, the $q$-boson-Toda correspondence in question yields Bethe Ansatz eigenfunctions in terms of hyperoctahedral Hall-Littlewood polynomials and provides the pertinent solutions of the Bethe Ansatz equations via  the global minima of corresponding Yang-Yang type Morse functions.
\end{abstract}

\maketitle

\section{Introduction}\label{sec1}
The relativistic Toda chain is an ubiquitous one-dimensional  $n$-particle model introduced by Ruijsenaars  that is integrable both at the level of classical and quantum mechanics
\cite{rui:relativistic}.  In the case of an open chain, integrable perturbations at the boundary were implemented
via the boundary Yang-Baxter equation \cite{kuz-tsy:quantum,sur:discrete}. At the quantum level, the hamiltonian of the relativistic Toda chain is given by a ($q$-)difference operator. Quantum groups connect the difference operator at issue to the quantum $K$-theory of flag manifolds \cite{bra-fin:finite,giv-lee:quantum} and
provide a natural representation-theoretical habitat for the construction of its eigenfunctions
\cite{eti:whittaker,sev:quantum}. 

When considering the quantum dynamics on an integer lattice the eigenvalue problem for the $q$-difference Toda chain can be solved in terms of $q$-Whittaker functions that arise as a
parameter specialization of the Macdonald polynomials, both in the case of particles moving on an infinite lattice \cite{ger-leb-obl:q-deformed} and in the case of  particles moving on a finite periodic lattice
\cite{duv-pas:q-bosons}. From the perspective of integrable probability, such particle models  are of interest in connection with the $q$-Whittaker process \cite{bor-cor:macdonald}. Representation-theoretical constructions for the pertinent   $q$-Whittaker functions can be found in
\cite{dif-ked:difference,dif-ked-tur:path,fei:fermionic}.

This note addresses the spectral problem for an open $n$-particle $q$-difference Toda chain on the  finite lattice $\{ 0,1,2,\ldots ,m\}$ that is
endowed with two-parameter  boundary interactions on both ends. The model could be thought of as a  finite discrete  and $q$-deformed counterpart of Sklyanin's open quantum Toda chain with general two-sided boundary perturbations governed by Morse potentials  \cite{skl:boundary}.
In the limit $m\to \infty$  the pertinent $q$-difference Toda hamiltonian was diagonalized
in terms of hyperoctahedral $q$-Whittaker functions that arise in turn through a parameter specialization of the Macdonald-Koornwinder polynomials \cite{die-ems:integrable}.
Here it will be shown that for finite $m$ an explicit eigenbasis  can be constructed from Bethe Ansatz wave functions given by Macdonald's hyperoctahedral Hall-Littlewood polynomials
\cite{mac:orthogonal}. 
The main idea is to exploit an equivalence between $q$-difference Toda chains and  $q$-boson models pointed out in \cite{duv-pas:q-bosons}.
By establishing a version of this equivalence in the current situation of an open chain with boundary perturbations, our $q$-difference Toda hamiltonian is mapped to the hamiltonian of
a $q$-boson model previously diagonalized in
\cite{die-ems-zur:completeness}. The upshot is that the commuting quantum integrals and the Bethe Ansatz eigenfunctions for  the  $q$-difference Toda chain can in this approach be retrieved
 directly from those in  \cite{die-ems-zur:completeness} for the corresponding $q$-boson model.

The material is organized as follows.
Section \ref{sec2} describes the hamiltonian of our $q$-difference Toda chain and verifies its self-adjointness. 
Section \ref{sec3} establishes  the equivalence with the $q$-boson model from  \cite{die-ems-zur:completeness}  and therewith retrieves the corresponding Bethe Ansatz wave functions
 in terms of hyperoctahedral Hall-Littlewood polynomials. 
The  Bethe Ansatz equations of interest are of a convex type studied in wider generality in
 \cite{die-ems:solutions}, which entails an explicit description of the spectrum via the global minima of associated  Yang-Yang-type Morse functions  detailed in Section \ref{sec4}.
The presentation closes in Section \ref{sec5} with a description of the spectral  analysis for the $q$-difference Toda chain in the degenerate  limit 
$q\to 1$.

\section{Open $q$-difference Toda chain with boundary interactions}\label{sec2}

\subsection{Quantum hamiltonian}
Given $m,n\in\mathbb{N}$, the $q$-difference Toda chain under consideration describes the quantum dynamics of $n$ interacting
particles hopping over the  finite integer lattice
 $\{ 0,1,2,\ldots , m \}$. The positions of these particles are encoded by a partition
 $\mu=(\mu_1,\mu_2,\ldots,\mu_n)$ in the configuration space
 \begin{equation}
 \Lambda^{(n,m)}=\{ \mu\in\mathbb{Z}^n \mid m\geq \mu_1\geq \mu_2\geq\cdots\geq\mu_n\geq0\} .
 \end{equation}
 The dynamics  is governed in turn by the following quantum hamiltonian 
% \begin{subequations}
 \begin{align}\label{Ht}
 H &= \beta_+ (1-q^{m-\mu_1})+\beta_-(1-q^{\mu_n} )\\
 &+  \sum_{1\leq  i\leq n} (1-\alpha_+ q^{m-\mu_1-1})^{\delta_{i-1}}(1-q^{\mu_{i-1}-\mu_i})T_i \nonumber  \\
& + \sum_{1\leq i\leq n} (1-\alpha_- q^{\mu_n-1})^{\delta_{n-i}}(1-q^{\mu_i-\mu_{i+1}})T_i^{-1} ,  \nonumber
 \end{align}
with
 \begin{equation*}\label{virtual}
\delta_i=\begin{cases}
1&\text{if}\ i=0, \\
0&\text{if}\ i\neq 0,
\end{cases}\qquad \qquad  \begin{cases} \mu_0\equiv m, \\ \mu_{n+1}\equiv 0. \end{cases}
\end{equation*}
%\end{subequations}
Here $T_i$ and $T_i^{-1}$ denote hopping operators that act on  $n$-particle wave functions $\psi$ via a unit translation of the $i$th particle to the left and to the right, respectively:
 \begin{equation*}
 (T_i^{\epsilon} \psi) (\mu_1,\ldots,\mu_n)=  \psi (\mu_1,\ldots,\mu_{i-1},\mu_i+\epsilon,\mu_{i+1},\ldots, \mu_n) \quad(\epsilon\in \{1,-1\}).
 \end{equation*}
 
 The action of $H$ \eqref{Ht} on  wave functions
 $\psi:\Lambda^{(n,m)}\to\mathbb{C}$ is well-defined in the sense that the coefficient of 
 $ (T_i^{\epsilon} \psi) (\mu_1,\ldots,\mu_n)$ in $(H\psi)(\mu_1,\ldots,\mu_n)$
 vanishes for any $(\mu_1,\ldots,\mu_n)\in\Lambda^{(n,m)}$ such that
 $ (\mu_1,\ldots,\mu_{i-1},\mu_i+\epsilon,\mu_{i+1},\ldots, \mu_n) \not\in\Lambda^{(n,m)}$.
Notice also that the convention in the second brace below Eq. \eqref{Ht} can be interpreted as representing the positions of two additional particles fixed at the lattice end-points $0$ and $m$, respectively.  The parameter $q\in (-1,1)\setminus \{ 0\}$ denotes a scale parameter of the model governing the nearest neighbour interaction between the particles whereas the parameters
 $\alpha_\pm\in (-1, 1)$ and $\beta_\pm\in\mathbb{R}$ represent coupling constants regulating additional interactions at the boundary of the chain.
 
 Our main goal is to solve the spectral problem for the $q$-difference Toda hamiltonian $H$ \eqref{Ht} 
 in the $\binom{n+m}{n}$-dimensional Hilbert space $\ell^2\bigl(\Lambda^{(n,m)},\Delta \bigr)$
 of functions $\psi:\Lambda^{(n,m)}\to\mathbb{C}$, endowed with an inner product
 \begin{subequations}
\begin{equation}
\langle \psi,\phi\rangle_\Delta = \sum_{\mu\in\Lambda^{(n,m)} }\psi (\mu)\overline{\phi(\mu)}\Delta _\mu \qquad \left(\psi,\phi\in \ell^2\bigl(\Lambda^{(n,m)},\Delta \bigr) \right)
\end{equation}
determined by positive weights given by perturbed $q$-multinomials on $\Lambda^{(n,m)}$:
\begin{equation}\label{weights}
\Delta _\mu =
\frac{(q;q)_m}{(\alpha_+;q)_{m-\mu_1}  (\alpha_-;q)_{\mu_n} 
\prod_{0\leq i\leq n} (q;q)_{\mu_i-\mu_{i+1}} } 
\end{equation}
$ (\mu\in\Lambda^{(n,m)} $).
\end{subequations}
Here we have employed the standard $q$-shifted factorial
\begin{equation*}
(a;q)_l=
\begin{cases}
1&\text{if}\ l=0 \\
 (1-a)(1-aq)\cdots (1-aq^{l-1})&\text{if}\  l=1,2,3,\ldots
 \end{cases}
 \end{equation*}

\begin{proposition}[Self-adjointness]\label{sa:prp}
For   $\alpha_\pm\in (-1, 1)$,  $\beta_\pm\in\mathbb{R}$, and $q\in (-1,1)\setminus \{ 0\}$, the $q$-difference Toda hamiltonian $H$ \eqref{Ht} is self-adjoint in $\ell^2\bigl(\Lambda^{(n,m)},\Delta \bigr)$, i.e.
\begin{equation*}
\forall \psi,\phi\in \ell^2\bigl(\Lambda^{(n,m)},\Delta \bigr):\qquad \langle H\psi,\phi\rangle_\Delta = \langle \psi, H \phi\rangle_\Delta .
\end{equation*}
\end{proposition}

\begin{remark}
For $m\to\infty$, the $q$-difference Toda hamiltonian $H$ \eqref{Ht} was diagonalized  in \cite[Section 7]{die-ems:integrable} in terms of a unitary eigenfunction transform with a $q$-Whittaker  kernel built from a parameter specialization of the Macdonald-Koornwinder polynomials.
\end{remark}

\subsection{Proof of Proposition \ref{sa:prp}}
The action of $H$ \eqref{Ht} on $ \psi\in \ell^2\bigl(\Lambda^{(n,m)},\Delta \bigr)$ is of the form
\begin{align*}
(H\psi)(\mu)
=& \bigl( \beta_+ (1-q^{m-\mu_1}) +
\beta_- (1-q^{\mu_n}) \bigr)
 \psi(\mu ) \\
+&\sum_{\substack{1\leq i \leq n\\ \mu+e_i\in\Lambda^{(n,m)}}}
(1-\alpha_+ q^{m-\mu_1-1})^{\delta_{i-1} }
 (1-q^{\mu_{i-1}-\mu_i}) \psi (\mu+e_i )  \nonumber \\
+&\sum_{\substack{1\leq i \leq n\\ \mu -e_i\in\Lambda^{(n,m)}}}
(1-\alpha_-q^{\mu_n-1})^{\delta_{n-i} }
(1-q^{ \mu_i-\mu_{i+1}})  \psi (\mu-e_i) , \nonumber
\end{align*}
where  the vectors $e_1,\ldots, e_n$ represent the standard unit basis for $\mathbb{Z}^n$.

Since all coefficients of the difference operator in question are real, the asserted symmetry
 $\langle H\psi,\phi\rangle_\Delta = \langle \psi, H \phi\rangle_\Delta $ is immediate 
 from the following bilinear identity $\forall \psi,\phi\in  \ell^2\bigl(\Lambda^{(n,m)},\Delta \bigr)$:
 \begin{align*}
& \sum_{\mu\in\Lambda^{(n,m)}}
\Bigl( \sum_{\substack{1\leq i \leq n\\ \mu+e_i\in\Lambda^{(n,m)}}}
(1-\alpha_+ q^{m-\mu_1-1})^{\delta_{i-1} }
 (1-q^{\mu_{i-1}-\mu_i}) \psi (\mu+e_i ) \Bigr) \phi(\mu)  \Delta_\mu \\
  &\stackrel{(i)}{=}
  \sum_{\tilde\mu\in\Lambda^{(n,m)}}
\sum_{\substack{1\leq i \leq n\\ \tilde\mu-e_i\in\Lambda^{(n,m)}}}
(1-\alpha_+ q^{m-\tilde\mu_1})^{\delta_{i-1} }
 (1-q^{\tilde\mu_{i-1}-\tilde\mu_i+1}) \psi (\tilde\mu ) \phi(\tilde\mu-e_i) 
 \Delta_{\tilde\mu-e_i } \\
  &\stackrel{(ii)}{=}
  \sum_{\tilde\mu\in\Lambda^{(n,m)}}   \psi (\tilde\mu) 
 \Bigl( \sum_{\substack{1\leq i \leq n\\ \tilde\mu -e_i\in\Lambda^{(n,m)}}}
 (1-\alpha_-q^{\tilde\mu_n-1})^{\delta_{n-i} }
(1-q^{ \tilde\mu_i-\tilde\mu_{i+1}})
\phi(\tilde\mu-e_i)  \Bigr) \Delta_{\tilde\mu} .
\end{align*}
Step \emph{(i)} hinges on the substitution  $\mu=\tilde\mu-e_i$, which for a given 
$i\in\{ 1,\ldots ,n\}$ determines a bijection from the subset
$\{ \tilde\mu\in \Lambda^{(n,m)}\mid \tilde\mu-e_i\in \Lambda^{(n,m)}\}$  onto the subset
$\{ \mu\in \Lambda^{(n,m)}\mid \mu+e_i\in \Lambda^{(n,m)}\}$.
Step \emph{(ii)} uses
the elementary recurrence
$$
(1-\alpha_+ q^{m-\tilde\mu_1})^{\delta_{i-1} }
 (1-q^{\tilde\mu_{i-1}-\tilde\mu_i+1})  \Delta_{\tilde\mu-e_i }= (1-\alpha_-q^{\tilde\mu_n-1})^{\delta_{n-i} }
(1-q^{\tilde\mu_i-\tilde\mu_{i+1}})  \Delta_{\tilde\mu}
$$
for $\tilde\mu\in\Lambda^{(n,m)}$ such that $\tilde\mu-e_i\in \Lambda^{(n,m)}$ (and the convention $\tilde\mu_0\equiv m$, $\tilde\mu_{n+1}\equiv 0$).

\section{Eigenfunctions}\label{sec3}

\subsection{Bethe Ansatz}
While Proposition \ref{sa:prp} implies that the existence of an orthogonal eigenbasis diagonalizing $H$ \eqref{Ht}  in  $\ell^2\bigl(\Lambda^{(n,m)},\Delta \bigr)$ is evident from the spectral theorem for self-adjoint operators in finite dimension, the aim here is to provide an \emph{explicit} eigenbasis given by Bethe Ansatz wave functions in the spirit of \cite{die:q-deformation} for $n=1$.

To this end let us recall that for any $\lambda=(\lambda_1,\ldots ,\lambda_m)\in\Lambda^{(m,n)}$ and
$\xi=(\xi_1,\ldots, \xi_m)$ belonging to
\begin{equation}\label{regular}
\mathbb{R}^m_{\rm{reg}}= \{ \xi\in\mathbb{R}^m\mid 
2\xi_j,\xi_j-\xi_k,\xi_j+\xi_k\not\in 2\pi\mathbb{Z},\ \forall 1\leq j\neq k\leq m\} ,
\end{equation}
Macdonald's hyperoctahedral Hall-Littlewood polynomial (associated with the root system $BC_m$) is given by
\cite[\S 10]{mac:orthogonal}
\begin{subequations}
\begin{align}\label{HL}
R_\lambda (\xi_1,\ldots,\xi_m)& =   \\
 \sum_{\substack{ \sigma\in S_m \\ \epsilon\in \{ 1,-1\}^m}}  & C(\epsilon_1 \xi_{\sigma(1)},\ldots , \epsilon_m \xi_{\sigma(m)})
\exp ({\rm i}\epsilon_1 \xi_{\sigma(1)}\lambda_1+\cdots +{\rm i} \epsilon_n \xi_{\sigma(m)} \lambda_m)  \nonumber
\end{align}
with 
\begin{eqnarray}\label{C}
\lefteqn{C(\xi_1,\ldots ,\xi_m) =
\prod_{1\leq j\leq m} \frac{(1- \beta_+ e^{-{\rm i}\xi_j}+\alpha_+ e^{-2{\rm i}\xi_j})}{1-e^{-2{\rm i}\xi_j}}} && \\
&& \times \prod_{1\leq j<k \leq m} \left(\frac{1-q e^{-{\rm i}(\xi_j-\xi_k)}}{1-e^{-{\rm i}(\xi_j-\xi_k)}}\right)\left(  \frac{1-q e^{-{\rm i}(\xi_j+\xi_k)}}{1-e^{-{\rm i}(\xi_j+\xi_k)}} \right)  ,\nonumber
\end{eqnarray}
\end{subequations}
where the summation is over all permutations $\sigma= { \bigl( \begin{smallmatrix}1& 2& \cdots & m \\
 \sigma (1)&\sigma (2)&\cdots & \sigma (m)
 \end{smallmatrix}\bigr)}$ of the symmetric group $S_m$ and all sign configurations
$\epsilon=(\epsilon_1,\ldots,\epsilon_m)\in \{ 1,-1\}^m$.

For any $\lambda\in\Lambda^{(m,n)}$ and $0\leq i\leq n$ we denote the multiplicity of $i$ in $\lambda$ by 
\begin{subequations}
\begin{equation}
\mathrm{m}_i(\lambda) = |\{ 1\leq j\leq m \mid \lambda_{j}=i\} | .
\end{equation}
Additionally,  for $\mu\in\Lambda^{(n,m)}$ we write
\begin{equation}\label{conjugate}
\mu^\prime =\bigl(0^{m-\mu_1}1^{\mu_1-\mu_2}2^{\mu_2-\mu_3}\cdots
(n-1)^{\mu_{n-1}-\mu_n} n^{\mu_n}\bigr) 
\end{equation}
\end{subequations}
for  its conjugate partition  $\mu^\prime\in\Lambda^{(m,n)}$
(i.e. {`}with the columns and rows swapped'). In other words,
$\mu^\prime=(\mu_1^\prime,\mu_2^\prime,\ldots,\mu_m^\prime)$ is the (unique) partition in $\Lambda^{(m,n)}$ such that $\mathrm{m}_i(\mu^\prime)=\mu_i-\mu_{i+1}$ for $i=0,\ldots ,n$ (where, recall,  $\mu_0\equiv m$ and $\mu_{n+1}\equiv 0$). 
Notice in this connection that $| \Lambda^{(n,m)} | = | \Lambda^{(m,n)} | =\frac{(n+m)!}{n! \, m!}$ and that the mapping $\mu\to\mu^\prime$ \eqref{conjugate} defines a bijection from  $\Lambda^{(n,m)}$ onto  $\Lambda^{(m,n)}$.

From now on it will moreover be assumed (unless explicitly stated otherwise) that  the boundary parameters $\alpha_\pm$, $ \beta_\pm$ have values such that the roots $p_\pm$, $q_\pm$ of the two quadratic  polynomials $x^2-\beta_\pm x+\alpha_\pm$ belong to the interval $(-1,1)\setminus \{ 0\}$:
\begin{subequations}
\begin{equation}\label{bp:a}
\boxed{\alpha_\pm = p_\pm q_\pm \quad \text{and}\quad \beta_\pm=p_\pm+q_\pm\quad\text{with}\ q_\pm,p_\pm\in (-1,1) \setminus \{ 0\}}
\end{equation}
or equivalently
\begin{equation}\label{bp:b}
\boxed{0<\alpha_\pm^2 <1 \quad \text{and}\quad 4\alpha_\pm\leq \beta_\pm^2\leq (1+\alpha_\pm)^2 .}
\end{equation}
\end{subequations}

\begin{theorem}[Bethe Ansatz Wave Function]\label{BA:thm}
Let $q\in (-1,1)\setminus \{0\}$ and let the boundary parameters $\alpha_\pm$, $\beta_\pm$ belong to the domain specified in Eqs. \eqref{bp:a}, \eqref{bp:b}.
Given
$\xi=(\xi_1,\ldots ,\xi_m)\in\mathbb{R}^m_{\rm{reg}}$ \eqref{regular}, we define the wave function 
 $\psi_\xi \in \ell^2\bigl(\Lambda^{(n,m)},\Delta \bigr)$
through its values on $\Lambda^{(n,m)}$ as follows:
\begin{subequations}
\begin{equation}\label{BA}
\psi_\xi (\mu)=R_{\mu^\prime}(\xi_1,\ldots,\xi_m)\qquad (\mu\in\Lambda^{(n,m)}).
\end{equation}

The wave function $\psi_\xi$ \eqref{BA} solves the eigenvalue equation for the $q$-difference Toda hamiltonian $H$ \eqref{Ht}
\begin{equation}\label{ev-eq}
H\psi_\xi= E(\xi) \psi_\xi\quad \text{with}\quad E(\xi)=2(1-q)\sum_{1\leq j\leq m} \cos (\xi_j),
\end{equation}
provided the spectral parameter $\xi\in\mathbb{R}^m_{\rm reg}$ satisfies the algebraic system of Bethe Ansatz equations
\begin{align}\label{BAE}
 e^{2 {\rm i} n \xi_j}
&= \frac{ (1- \beta_+ e^{{\rm i}\xi_j}+\alpha_+ e^{2{\rm i}\xi_j})}{ (e^{2{\rm i}\xi_j}- \beta_+ e^{{\rm i}\xi_j}+\alpha_+ ) } 
  \frac{ (1- \beta_- e^{{\rm i}\xi_j}+\alpha_- e^{2{\rm i}\xi_j}) } {  (e^{2{\rm i}\xi_j}- \beta_- e^{{\rm i}\xi_j}+\alpha_- ) } \\
 &\qquad \times \prod_{\substack{1\le k \le m \\ k\neq j}}
       \frac{ (1-q e^{{\rm i}(\xi_j - \xi_k)}) (1-q e^{{\rm i}(\xi_j + \xi_k)})   }
             { ( e^{{\rm i}(\xi_j - \xi_k)}-q ) ( e^{{\rm i}(\xi_j+ \xi_k)}-q )   } \qquad\text{for}\quad j=1,\ldots ,m . \nonumber
\end{align}
\end{subequations}
\end{theorem}

\begin{remark}
The hyperoctahedral Hall-Littlewood polynomial $R_{\lambda}(\xi_1,\ldots,\xi_m)$, $\lambda\in\Lambda^{(m,n)}$ is  in fact a  symmetric polynomial in  $\cos(\xi_1),\ldots,\cos(\xi_m)$
of total degree $\lambda_1+\lambda_2+\cdots+\lambda_n$. Hence, it is clear that the Bethe Ansatz wave function $\psi_\xi$ in Theorem \ref{BA:thm} extends smoothly in the spectral parameter $\xi$ from values in
$\mathbb{R}^m_{\rm{reg}}$ to values in $\mathbb{R}^m$.
\end{remark}

\subsection{Proof of Theorem \ref{BA:thm}}
In \cite{die-ems-zur:completeness}, 
Propositions 8.4, 9.2 and (the proof of) Theorem 9.5 imply that the hyperoctahedral Hall-Littlewood polynomials satisfy the following  recurrence for $\lambda\in\Lambda^{(m,n)}$:
\begin{align}\label{rec-rel}
E(\xi) R_\lambda (\xi)
=& \bigl( \beta_+ (1-q^{\mathrm{m}_0(\lambda)}) +
\beta_- (1-q^{\mathrm{m}_n(\lambda)}) \bigr)
 R_\lambda (\xi) \\
+&\sum_{\substack{1\leq j \leq m\\ \lambda+e_j\in\Lambda^{(m,n)}}}
(1-\alpha_+ q^{\mathrm{m}_0(\lambda)-1})^{\delta_{\lambda_j} }
 (1-q^{\mathrm{m}_{\lambda_j}(\lambda)}) R_{\lambda+e_j} (\xi)  \nonumber \\
+&\sum_{\substack{1\leq j \leq m\\ \lambda-e_j\in\Lambda^{(m,n)}}}
(1-\alpha_-q^{\mathrm{m}_n(\lambda)-1})^{\delta_{n-\lambda_j} }
(1-q^{ \mathrm{m}_{\lambda_j}(\lambda)})  R_{\lambda-e_j} (\xi) ,  \nonumber
\end{align}
which is on-shell in the sense that the relation holds provided $\xi=(\xi_1,\ldots,\xi_m)$ satisfies the Bethe Ansatz equations \eqref{BAE}.
Substituting $\lambda=\mu^\prime$ with $\mu\in\Lambda^{(n,m)}$ leads us via Eq. \eqref{conjugate} to the recurrence relation
\begin{align*}
E(\xi) R_{\mu^\prime} (\xi)
=& \bigl( \beta_+ (1-q^{m-\mu_1}) +
\beta_- (1-q^{\mu_n}) \bigr)
 R_{\mu^\prime} (\xi) \\
+&\sum_{\substack{1\leq j \leq m\\ \mu^\prime+e_j\in\Lambda^{(m,n)}}}
(1-\alpha_+ q^{m-\mu_1-1})^{\delta_{\mu^\prime_j} }
 (1-q^{\mathrm{m}_{\mu^\prime_j}(\mu^\prime)}) R_{\mu^\prime +e_j} (\xi)  \nonumber \\
+&\sum_{\substack{1\leq j \leq m\\ \mu^\prime-e_j\in\Lambda^{(m,n,)}}}
(1-\alpha_-q^{\mu_n-1})^{\delta_{n-\mu^\prime_j} }
(1-q^{ \mathrm{m}_{\mu^\prime_j}(\mu^\prime)})  R_{\mu^\prime-e_j} (\xi) .
\end{align*}
We now observe that for any  $\mu\in\Lambda^{(n,m)}$ and $j\in \{ 1,\ldots ,m\}$:
\begin{equation*}
\mu^\prime +e_j\in\Lambda^{(m,n)} \Longleftrightarrow \mu^\prime +e_j=(\mu+e_i)^\prime
\quad\text{with}\ i=\mu_j^\prime +1 \in \{ 1,\ldots,n \}
\end{equation*}
and
\begin{equation*}
\mu^\prime -e_j\in\Lambda^{(m,n)} \Longleftrightarrow \mu^\prime -e_j=(\mu-e_i)^\prime
\quad\text{with}\ i=\mu_j^\prime \in \{ 1,\ldots,n \} .
\end{equation*}
The recurrence of interest can thus be rewritten in the form 
\begin{align*}
E(\xi) R_{\mu^\prime} (\xi)
=& \bigl( \beta_+ (1-q^{m-\mu_1}) +
\beta_- (1-q^{\mu_n}) \bigr)
 R_{\mu^\prime} (\xi) \\
+&\sum_{\substack{1\leq i \leq n\\ \mu+e_i\in\Lambda^{(n,m)}}}
(1-\alpha_+ q^{m-\mu_1-1})^{\delta_{i-1} }
 (1-q^{\mu_{i-1}-\mu_i}) R_{(\mu+e_i)^\prime } (\xi)  \nonumber \\
+&\sum_{\substack{1\leq i \leq n\\ \mu -e_i\in\Lambda^{(n,m)}}}
(1-\alpha_-q^{\mu_n-1})^{\delta_{n-i} }
(1-q^{ \mu_i-\mu_{i+1}})  R_{(\mu-e_i)^\prime} (\xi) ,
\end{align*}
where we have employed  once more  Eq. \eqref{conjugate}.

\section{Spectral analysis}\label{sec4}

\subsection{Solutions for the Bethe Ansatz equations}
The following system of transcendental  equations provides a logarithmic form of the Bethe Ansatz equations in Theorem \ref{BA:thm}:
\begin{subequations}
\begin{align}\label{critical-eq}
2n\xi_j+v_{p_+}(\xi_j)& +v_{q_+}(\xi_j)+v_{p_-}(\xi_j)+v_{q_-}(\xi_j)\\
 &+\sum_{\substack{1\leq k\leq m\\k\neq j}} \Bigl( v_q(\xi_j+\xi_k)+ v_q(\xi_j-\xi_k)\Bigr)=2\pi \bigl(m+1-j+\kappa_j\bigr)  ,
 \nonumber
\end{align}
 with $j=1,\ldots, m$, $\kappa\in \Lambda^{(m,n)}$ and
\begin{equation}\label{v}
v_a( z) = 
\int_0^z \frac{ (1-a^2)\ \text{d}x}{1-2a\cos(x)+a^2}
=
{\rm i} \log
\biggl( \frac{1- ae^{{\rm i} z}}{e^{{\rm i}z} - a }  \biggr) \qquad (-1<a<1).
\end{equation}
\end{subequations}
Indeed, upon multiplying  Eq. \eqref{critical-eq}  by ${\rm i}$ ($=\sqrt{-1}$) and applying the exponential  function on  both sides it is readily  seen that any of its solutions
gives rise to a solution of the Bethe Ansatz equations \eqref{BAE}
 (where---recall---the boundary parameters $\alpha_\pm$, $\beta_\pm$ and $p_\pm$, $q_\pm$ are related via Eq. \eqref{bp:a}).

For any $\kappa\in\Lambda^{(m,n)}$, the system in Eqs. \eqref{critical-eq}, \eqref{v} describes the critical point of a Yang-Yang type Morse function:
{\small
\begin{align}\label{Morse}
&V^{}_\kappa(\xi_1,\ldots,\xi_m)= \sum_{1\le j < k \le m }
\left(
    \int_0^{\xi_j+\xi_k} v_q(x)\text{d}x+   \int_0^{\xi_j-\xi_k} v_q(x)\text{d}x
\right) +  \\
&
\sum_{1\leq j\leq m}  \left(
n\xi_j^2-2\pi \bigl( m+1-j+\kappa_j\bigr) \xi_j 
+   \int^{\xi_j}_0 \bigl(v_{p_+}(x)+v_{q_+}(x)+v_{p_-}(x)+v_{q_-}(x)\bigr)\text{d}x\right)  .\nonumber
\end{align}
}
The function $V^{}_\kappa(\xi_1,\ldots,\xi_m)$ belongs to a wider class of smooth, strictly convex and radially unbounded Morse functions studied in \cite[Section 3]{die-ems:solutions}. 
The upshot is that via a Yang-Yang type analysis, one arrives at $\binom{m+n}{n}$ solutions for the Bethe Ansatz  equations given by the respective minima
of $V_\kappa (\xi_,\ldots,\xi_m)$, $\kappa\in\Lambda^{(m,n)}$ (cf. \cite[Remark 3.5]{die-ems-zur:completeness}).

\begin{proposition}[Solutions for the Bethe Ansatz Equations]\label{BAE-sol:prp}
Let $q\in (-1,1)\setminus \{0\}$ and let the boundary parameters $\alpha_\pm$, $\beta_\pm$ belong to the domain specified in Eqs. \eqref{bp:a}, \eqref{bp:b}.

(i) For any $\kappa\in\Lambda^{(m,n)}$, the logarithmic form of the Bethe Ansatz equations in Eqs. \eqref{critical-eq}, \eqref{v}
has a unique solution $\boxed{\xi_\kappa\in\mathbb{R}^m}$ given by the global minimum of the strictly convex radially unbounded Morse function $V_\kappa (\xi_1,\ldots \xi_m)$ \eqref{Morse}.

(ii) The global minima $\xi_\kappa$, $\kappa\in\Lambda^{(m,n)}$ in part (i) are all distinct and located within the open alcove
\begin{subequations}
\begin{equation}\label{alcove}
 \mathbb{A}^{m}=\{ (\xi_1,\xi_2,\ldots,\xi_m)\in\mathbb{R}^m\mid \pi>\xi_1>\xi_2>\cdots >\xi_m>0\}  \subset
 \mathbb{R}^m_{\rm{reg}} .
\end{equation}
Moreover, at a global minimum $\xi=\xi_\kappa$ the following estimates are fulfilled:
\begin{equation}\label{bounds:a}
\frac{\pi(m+1-j+\kappa_j)}{n+\textsc{k}_+} \leq \xi_j \leq \frac{\pi(m+1-j+\kappa_j)}{n+\textsc{k}_-}
\end{equation}
(for $1\leq j\leq m$),  and
\begin{equation}\label{bounds:b}
\frac{\pi(k-j+\kappa_j-\kappa_k)}{n+\textsc{k}_+} \leq \xi_j-\xi_k\leq \frac{\pi(k-j+\kappa_j-\kappa_k)}{n+\textsc{k}_-} 
\end{equation}
(for $1 \le  j < k \le m$), where
\begin{align}
\textsc{k}_\pm = & (m-1)\left( \frac{1+|q|}{1 - |q|}\right)^{\pm 1} + \\
&\frac{1}{2}\left(  \left(\frac{1+|p_+|}{1-|p_+|}\right)^{\pm 1}+\left(\frac{1+|q_+|}{1- |q_+| }\right)^{\pm 1}+ \left(\frac{1+|p_- |}{1- |p_-|}\right)^{\pm 1}+\left(\frac{1-|q_-|}{1- |q_-|} \right)^{\pm 1}\right).
\nonumber
\end{align}
\end{subequations}
\end{proposition}

\begin{proof}
The assertions of this proposition follow by applying \cite[Propositions 3.1 and 3.2]{die-ems:solutions} to the Bethe Ansatz equations of Theorem \ref{BA:thm}  (cf. \cite[Remark 3.5]{die-ems-zur:completeness}).
\end{proof}

\begin{remark}\label{gradient:rem}
From a mostly academic perspective, Proposition \ref{BAE-sol:prp} invites us to compute $\xi_\kappa$ via the gradient flow of the pertinent Morse function:
\begin{equation}
\frac{\text{d}\xi_j}{\text{d}t}+\partial_{\xi_j} V_\kappa (\xi_1,\ldots ,\xi_m)=0, \qquad j=1,\ldots ,m,
\end{equation}
which gives rise to the following system of differential equations (cf. Eq. \eqref{critical-eq})
\begin{align}\label{g-flow}
\frac{\text{d}\xi_j}{\text{d}t}+2n\xi_j+v_{p_+}(\xi_j)& +v_{q_+}(\xi_j)+v_{p_-}(\xi_j)+v_{q_-}(\xi_j)\\
 &+\sum_{\substack{1\leq k\leq m\\k\neq j}} \Bigl( v_q(\xi_j+\xi_k)+ v_q(\xi_j-\xi_k)\Bigr)=2\pi \bigl(m+1-j+\kappa_j\bigr)  ,
 \nonumber
\end{align}
$j=1,\ldots ,m$.  For $n=0$ and $\kappa =(0^m)$, the corresponding gradient flow was analyzed in \cite{die:gradient}
and seen to converge exponentially fast to the roots of the Askey-Wilson polynomial $p_{m}(\cos \vartheta  ; p_+,q_+p_-, q_-| \rm{q}) $ \cite[Equation (14.1.1)]{koe-les-swa:hypergeometric} located within the interval  of orthogonality $0<\vartheta<\pi$.
A minor variation of \cite[Theorem 2]{die:gradient} reveals that in our present setting the equilibrium $\xi_\kappa$ of the gradient system in Eq.  \eqref{g-flow} remains globally exponentially stable, i.e. for any initial condition $\xi_\kappa (0)$ the unique solution $\xi_\kappa (t)$, $t\geq 0$ of the gradient system converges exponentially fast to the equilibrium $\xi_\kappa$.
More specifically, by slightly adapting the analisis in \cite[Section 4]{die:gradient} one readily deduces that for any $0<\varepsilon <2(n+\textsc{k}_-)$ ($=$ a lower bound for the eigenvalues of the hessian of $V_\kappa(\xi_1,\ldots,\xi_m)$), there exists a constant $C_\varepsilon>0$  such that
\begin{equation}\label{estimate}
\forall t\geq 0:\quad \|\xi_\kappa (t)-\xi_\kappa\|_\infty     \leq  C_\varepsilon e^{-\varepsilon t} ,
\end{equation}
where $\| \xi \|_\infty\equiv \max_{1\leq j\leq m} |\xi_j|$.   Apart from $\varepsilon$, the actual value of the constant $C_\varepsilon$  in the uniform estimate of the error term will depend on the choice of the initial condition $\xi_\kappa (0)\in\mathbb{R}^m$, as well as on
$\kappa\in\Lambda^{(m,n)}$ and $q,p_\pm,q_\pm\in (-1,1)$ (cf. Remark \ref{smooth:rem} below).  Notice that the $q$-difference Toda hamiltonian $H$ \eqref{Ht} degenerates to a discrete Laplacian on $\Lambda^{(n,m)}$ in the symplectic Schur limit $\alpha_\pm,\beta_\pm, q\to 0$.   At this elementary point in the parameters space, one has that
(cf. Eq. \eqref{bounds:a})
\begin{equation*}
{\textstyle \xi_\kappa \to \left(  \frac{\pi (m+\kappa_1) }{m+n+1} ,   \frac{\pi (m-1+\kappa_2) }{m+n+1} ,\ldots , \frac{\pi (m+1-j+\kappa_j) }{m+n+1} ,\ldots, \frac{\pi (1+\kappa_m) }{m+n+1} \right) } ,
\end{equation*}
which serves as a convenient initial condition for the gradient flow \eqref{g-flow}. Indeed, at this  particular value of the spectral parameter the bounds in Eqs. \eqref{bounds:a} and \eqref{bounds:b} are fulfilled  for any $p_\pm,q_\pm, q\in (-1,1)$.
\end{remark}

\subsection{Spectrum and eigenbasis}
By combining Theorem \ref{BA:thm} and Proposition \ref{BAE-sol:prp},  an eigenbasis  of Bethe Ansatz wave functions  for the $q$-difference Toda hamiltonian $H$ \eqref{Ht} in the Hilbert space  $\ell (\Lambda^{(n,m)},\Delta)$  is found together with the corresponding eigenvalues.

\begin{theorem}[Spectrum and Eigenbasis]\label{diagonalization:thm}
Let $q\in (-1,1)\setminus \{0\}$ and let the boundary parameters $\alpha_\pm$, $\beta_\pm$ belong to the domain specified in Eqs. \eqref{bp:a}, \eqref{bp:b}. For any $\xi\in\mathbb{R}^m_{\rm reg}$ and
$\kappa\in\Lambda^{(m,n)}$,  the function $\psi_\xi :\Lambda^{(n,m)}\to \mathbb{R}$  refers to the Hall-Littlewood Bethe Ansatz wave function from Theorem \ref{BA:thm} and
 $\xi_\kappa\in\mathbb{A}^m\subset \mathbb{R}^m_{\rm reg}$ denotes the unique global minimum of $V_\kappa (\xi_1,\ldots ,\xi_m)$ \eqref{Morse}
detailed in Proposition \ref{BAE-sol:prp}.

(i) The spectrum of the $q$-difference Toda hamiltonian $H$ \eqref{Ht} in the Hilbert space $\ell (\Lambda^{(n,m)},\Delta)$ consists of the eigenvalues $E(\xi_\kappa)$, $\kappa\in\Lambda^{(m,n)}$, where
$E(\xi)$ is given by Eq. \eqref{ev-eq}.

(ii) The corresponding Bethe Ansatz wave functions $\psi_{\xi_{\kappa}}$, $\kappa\in\Lambda^{(m,n)}$ constitute an eigenbasis for
 $H$ \eqref{Ht}  in $\ell (\Lambda^{(n,m)},\Delta)$ such that
\begin{equation}\label{diagonal}
H\psi_{\xi_\kappa}= E(\xi_\kappa) \psi_{\xi_\kappa}\qquad (\kappa\in\Lambda^{(m,n)}).
\end{equation}
\end{theorem}

\begin{proof}
It is clear from Theorem \ref{BA:thm} and Proposition \ref{BAE-sol:prp} that for any $\kappa\in\Lambda^{(m,n)}$  the Bethe Ansatz wave function
$\psi_{\xi_{\kappa}}$ solves the eigenvalue equation \eqref{diagonal}. Moreover, the value of the Bethe Ansatz wave functions at the origin is given by Macdonald's three-parameter Poincaré series for the root system
$BC_m$ \cite[\S 10]{mac:orthogonal}:
\begin{align*}
\psi_{\xi}(0^n)&= R_{(0^m)}(\xi_1,\ldots,\xi_m)=
 \sum_{\substack{ \sigma\in S_m \\ \epsilon\in \{ 1,-1\}^m}}  C(\epsilon_1 \xi_{\sigma (1)},\ldots , \epsilon_m \xi_{\sigma (m)}) \\
&=
\frac{(\alpha_+;q)_m (q;q)_m}{(1-q)^m} \neq 0
\end{align*}
(cf. also \cite[Remark 3.4]{die-ems-zur:completeness}). Hence, for any $\kappa\in\Lambda^{(m,n)}$ the wave function
$\psi_{\xi_{\kappa}}$ constitutes a proper (i.e. nontrivial) eigenfunction of the $q$-difference Toda hamiltonian $H$ with eigenvalue $E(\xi_\kappa)$. To confirm the completeness of the Bethe Ansatz it remains to check that the wave functions in question indeed form a basis for  $\ell (\Lambda^{(n,m)},\Delta)$, or equivalently, that the hyperoctahedral Hall-Littlewood polynomials $R_{\mu^\prime}(\xi_1,\ldots,\xi_m)$, $\mu\in\Lambda^{(n,m)}$ are linearly independent as functions on the Bethe spectrum   $\{ \xi_\kappa\mid \kappa\in\Lambda^{(m,n)}\}$.
This  independence  is immediate from the second part of  \cite[Theorem 3.1]{die-ems-zur:completeness}.
\end{proof}

\begin{remark}\label{smooth:rem}
The $q$-difference Toda hamiltonian $H$, the positive weights $\Delta_\mu$, and the Bethe Ansatz wave function $\psi_\xi$ clearly extend
 smoothly in the parameters to the domain
$q,p_\pm, q_\pm\in (-1,1)$.   With the aid of the implicit function theorem, it is seen that the same is true for the solutions  $\xi_\kappa$, $\kappa\in\Lambda^{(n,m)}$ of the (logarithmic)  Bethe Ansatz equation in Proposition \ref{BAE-sol:prp} (cf. \cite[Remark 3.6]{die-ems-zur:completeness}).
Indeed, the Morse function $V_\kappa (\xi_1,\ldots,\xi_m)$ extends smoothly in the parameters and remains strictly convex.
 Hence, for
$\kappa\in\Lambda^{(m,n)}$ the Bethe Ansatz wave function $\psi_{\xi_{\kappa}}$ constitutes in fact 
 an eigenfunction of $H$ \eqref{Ht} in  $\ell (\Lambda^{(n,m)},\Delta)$ with eigenvalue $E(\xi_\kappa)$ \eqref{ev-eq}  for any $q,p_\pm, q_\pm\in (-1,1)$ (i.e. even if one or more of the parameters in question vanish).
\end{remark}

\begin{remark}\label{orthogonality:rem}
The self-adjointness of $H$ \eqref{Ht} in Proposition \ref{sa:prp}  implies that
\begin{subequations}
\begin{equation}\label{orthogonality}
\forall \kappa,\nu\in\Lambda^{(m,n)}:\qquad \langle \psi_{\xi_\kappa}, \psi_{\xi_\nu}\rangle_\Delta = 0\quad \text{if}\ \kappa\neq \nu,
\end{equation}
provided $E(\xi_{\kappa})\neq E(\xi_{\nu})$.  Rewritten in terms of  the (real-valued) hyperoctahedral Hall-Littlewood polynomials 
$R_{\lambda} (\xi)=R_{\lambda} (\xi_1,\ldots ,\xi_m)$ one obtains that in this situation:
\begin{equation}\label{R-orthogonality}
 \sum_{\lambda \in\Lambda^{(m,n)} } R_{\lambda} (\xi_\kappa) R_{\lambda} (\xi_\nu ) \Delta^{\prime}_{\lambda } = 0 \quad \text{if}\ \kappa\neq \nu ,
\end{equation}
with
\begin{equation}\label{weights-conjugate}
\Delta^{\prime}_{\lambda } =
\frac{(q;q)_m}{(\alpha_+;q)_{\mathrm{m}_0(\lambda)} (\alpha_-;q)_{\mathrm{m}_n(\lambda)}  
\prod_{0\leq i\leq n} (q;q)_{\mathrm{m}_i (\lambda)} } 
\end{equation}
(so $\Delta^{\prime}_{\mu^\prime }  =\Delta_\mu$ for $\mu\in\Lambda^{(n,m)}$).
\end{subequations}
To date the latter orthogonality relation has been checked directly without the proviso regarding the nondegneracy of the corresponding eigenvalues of $H$ in the following four cases:
\begin{itemize}
\item[(a)]  if $q=0$ and $p_\pm,q_\pm\in (-1,1)$, cf. \cite[Theorem 3.1]{die-ems:discrete};
\item[(b)] if $0<q<1$, $\alpha_\pm =0$ and $\beta_\pm\in (-1,1)$, cf. \cite[Section 11.4]{die-ems:orthogonality};
\item[(c)] if $q,p_\pm, q_\pm\in (-1,1)$ and $n\geq 2m$, cf. \cite[Theorem 4.2]{die:harmonic};
\item[(d)] if $q,p_\pm, q_\pm\in (-1,1)$ and $n=1$, cf. \cite[Section 5.7]{die:harmonic}.
\end{itemize}
In view of Remark \ref{smooth:rem},  within these four subdomains the statements concerning the spectrum and completeness formulated in parts (i) and (ii) of Theorem \ref{diagonalization:thm} therefore persist with the corresponding Bethe Ansatz eigenbasis being orthogonal in  $\ell (\Lambda^{(n,m)},\Delta)$ (as expected).
\end{remark}

\begin{remark}
In \cite{die-ems-zur:completeness}, Eq. \eqref{rec-rel} is interpreted as the eigenvalue equation for a hamiltonian of an $m$-particle $q$-boson model on the lattice $\{ 0,1,\ldots ,n\}$.
The quantum integrability of this $m$-particle $q$-boson hamiltonian, which is thus given explicitly by the difference operator acting at the RHS of Eq. \eqref{rec-rel}, was established for $\alpha_+=\alpha_-=0$ in \cite{die-ems:orthogonality} (using the quantum inverse scattering method) and for general boundary parameters in  \cite[Section 8]{die-ems-zur:completeness}  (using representations of the double affine Hecke algebra of type $C^\vee C$ at the critical level $\rm{q}=0$).  As detailed explicitly for the hamiltonian in the proof of Theorem \ref{BA:thm}, the mapping $\mu\to\mu^\prime$ from $\Lambda^{(n,m)}$ onto $\Lambda^{(m,n)}$ allows us to pull back the commuting quantum integrals for the $q$-boson model from
$\ell(\Lambda^{(m,n)},\Delta^\prime)$ to $\ell(\Lambda^{(n,m)},\Delta)$. 
This maps the commuting quantum integrals in question to an algebra of  commuting difference operators in  $\ell(\Lambda^{(n,m)},\Delta)$ containing $H$ \eqref{Ht}.
Since \cite[Theorem 9.5]{die-ems-zur:completeness} guarantees that the latter algebra of commuting difference operators is (Harish-Chandra-)isomorphic to the algebra of complex functions
on the joint spectrum  $\{ \xi_\kappa\mid \kappa\in\Lambda^{(m,n)}\}\subset \mathbb{A}^m$, this establishes the quantum integrability of our $q$-difference Toda hamiltonian  $H$ \eqref{Ht}.
\end{remark}

\section{The limit $q\to 1$}\label{sec5}

\subsection{Quantum hamiltonian}\label{sec5.1}
Upon dividing out an overall scaling factor $1-q$,  the $q$-difference Toda hamiltonian $H$ \eqref{Ht} 
degenerates for
$q\to 1$ to an elementary difference operator $\tilde{H}$ with linear coefficients:
 \begin{align}\label{H-r}
 \tilde{H} &= \beta_+ (m-\mu_1)+\beta_-(\mu_n ) + \\
 &  \sum_{1\leq  i\leq n} \Bigl( (1-\alpha_+)^{\delta_{i-1}}(\mu_{i-1}-\mu_i)T_i  +(1-\alpha_- )^{\delta_{n-i}}(\mu_i-\mu_{i+1})T_i^{-1} \Bigr)  .  \nonumber
 \end{align}
From Proposition \ref{sa:prp} it follows that---assuming $\alpha_\pm\in (-1,1)$ and $\beta_\pm\in\mathbb{R}$---this limiting quantum hamiltonian is self-adjoint in a Hilbert space
$\ell (\Lambda^{(n,m)},\tilde{\Delta})$ governed by the weights of
a two-parameter multinomial distribution on the partitions $\Lambda^{(n,m)}$:
\begin{subequations}
\begin{align}\label{weights-r}
\tilde\Delta _\mu &=\lim_{q\to 1} \Delta_\mu \\
&=
\frac{m!}{(1-\alpha_+)^{m-\mu_1}   (1-\alpha_-)^{\mu_n } 
\prod_{0\leq i\leq n} (\mu_i-\mu_{i+1})! } \nonumber \\
&= \mathcal{N }\cdot
 \frac{  m!  \prod_{0\leq i \leq n} \rho_i^{\mu_i-\mu_{i+1}}}{
\prod_{0\leq i\leq n} (\mu_i-\mu_{i+1})! }  , \nonumber
\end{align}
with
\begin{equation}
\rho_i=
\begin{cases}
\frac{(1-\alpha_+)^{-1}}{(n-1)+(1-\alpha_+)^{-1} +(1-\alpha_-)^{-1}}&\text{if}\ i=0,\\[1ex]
\frac{1}{(n-1)+(1-\alpha_+)^{-1} +(1-\alpha_-)^{-1}}&\text{if}\ 0< i<n,
\\[1ex]
\frac{(1-\alpha_-)^{-1}}{(n-1)+(1-\alpha_+)^{-1} +(1-\alpha_-)^{-1}}&\text{if}\ i=n,
\end{cases}
\end{equation}
and
\begin{equation}\label{norm-r}
\mathcal{N }=\sum_{\mu\in\Lambda^{(n,m)}}\tilde\Delta_\mu= \bigl((n-1)+(1-\alpha_+)^{-1} +(1-\alpha_-)^{-1}\bigr)^m .
\end{equation}
\end{subequations}
Notice in particular that $\rho_0=\rho_2=\cdots =\rho_n=\frac{1}{n+1}$ if $\alpha_+=\alpha_-=0$.

\subsection{Bethe Ansatz}
The Bethe Ansatz wave function $\psi_\xi$  \eqref{BA} degenerates in the limit $q\to 1$ to a wave function
$\tilde\psi_\xi:\Lambda^{(n,m)}\to\mathbb{C}$ with values 
\begin{subequations}
\begin{equation}\label{BAr}
\tilde\psi_\xi (\mu)=\tilde{R}_{\mu^\prime}(\xi_1,\ldots,\xi_m)
\end{equation}
that separate in terms of univariate ($BC_1$-type) Hall-Littlewood polynomials. Specifically, for any $\lambda\in\Lambda^{(m,n)}$ and $\xi\in\mathbb{R}^m_{\rm{reg}}$ one has  that
\begin{align}\label{HLr}
\tilde{R}_\lambda (\xi_1,\ldots,\xi_m)  &=   \lim_{q\to 1} R_\lambda (\xi_1,\ldots,\xi_m)\\
&= \sum_{\sigma\in S_m}  R_{\lambda_1} ( \xi_{\sigma (1)})  R_{\lambda_2}( \xi_{\sigma (2)}) \cdots  R_{\lambda_m}( \xi_{\sigma(m)})\nonumber
\end{align}
with
\begin{align}\label{HL1}
 R_{l}(\vartheta )=&\frac{(1- \beta_+ e^{-{\rm i}\vartheta}+\alpha_+ e^{-2{\rm i}\vartheta})}{1-e^{-2{\rm i}\vartheta}} \exp ({\rm i} l \vartheta ) +\\
 &
 \frac{(1- \beta_+ e^{{\rm i}\vartheta }+\alpha_+ e^{2{\rm i}\vartheta })}{1-e^{2{\rm i}\vartheta}} \exp (-{\rm i} l \vartheta ) 
 \nonumber
\end{align}
($l\in \{ 0,\ldots, n\}$, $\vartheta\not\in \pi\mathbb{Z}$).
\end{subequations}

Heuristically, for the Bethe Ansatz wave function $ \tilde{\psi}_\xi$ \eqref{BAr}--\eqref{HL1} to solve the $q\to 1$ eigenvalue equation
\begin{subequations}
\begin{equation}\label{ev-eq-r}
\tilde{H} \tilde{\psi}_\xi= \tilde{E}(\xi) \tilde{\psi}_\xi\quad \text{with}\quad \tilde{E(}\xi)=2\sum_{1\leq j\leq m} \cos (\xi_j),
\end{equation}
one expects the spectral parameter $\xi $  to be required to satisfy the following decoupled system of Bethe Ansatz equations arising from Eq. \eqref{BAE} in the limit $q\to 1$:
\begin{equation}\label{BAE-r}
 e^{2 {\rm i} n \xi_j}
= \frac{ (1- \beta_+ e^{{\rm i}\xi_j}+\alpha_+ e^{2{\rm i}\xi_j})}{ (e^{2{\rm i}\xi_j}- \beta_+ e^{{\rm i}\xi_j}+\alpha_+ ) } 
  \frac{ (1- \beta_- e^{{\rm i}\xi_j}+\alpha_- e^{2{\rm i}\xi_j}) } {  (e^{2{\rm i}\xi_j}- \beta_- e^{{\rm i}\xi_j}+\alpha_- ) } ,
\end{equation}
$j=1,\ldots,m$.
\end{subequations}

\subsection{Solutions for the Bethe Ansatz equations}
It is  illuminating to emphasize that for $p_+,q_+,p_-,q_-\in (-1,1)$ the decoupled  Bethe Ansatz equations in Eq. \eqref{BAE-r} can be conveniently solved in terms of the roots of the Askey-Wilson polynomial
$p_{n+1}(\cos \vartheta  ; p_+,q_+p_-, q_-| \rm{q}) $ \cite[Equation (14.1.1)]{koe-les-swa:hypergeometric}  at $\rm{q}=0$. Indeed, it is clear from the orthogonality relation  \cite[Equation (14.1.2)]{koe-les-swa:hypergeometric}  that 
at $\rm{q}=0$ the Askey-Wilson polynomials fall within a well-known class of orthogonal polynomials studied by  Bernstein and Szegö \cite[Chapter 2.6]{sze:orthogonal}.  The classical theory of 
Bernstein and Szegö tells us, moreover, that the polynomials in question can be written explicitly as follows (cf. e.g. \cite[Section 4.3]{die-ems:quadrature}):
\begin{align}
 &\frac{p_{n+1}(\cos \vartheta  ; p_+,q_+p_-, q_-|0)}{ (p_+q_+p_-q_- q^n;q)_{n+1}}  = \\
 &\frac{\prod_{\epsilon=\pm}(1- p_\epsilon e^{-{\rm i}\vartheta } )(1- q_\epsilon e^{-{\rm i}\vartheta } )}{1-e^{-2{\rm i}\vartheta}} \exp ({\rm i}(n+1)\vartheta ) \nonumber \\
 &
+\frac{\prod_{\epsilon=\pm}(1- p_\epsilon e^{{\rm i}\vartheta } )(1- q_\epsilon e^{{\rm i}x\vartheta} )}{1-e^{2{\rm i}\vartheta}} \exp (-{\rm i}(n+1)\vartheta ) 
 \nonumber
\end{align}
($\vartheta\not\in \pi\mathbb{Z}$).
This explicit formula reveals in particular that the roots
\begin{equation}\label{roots}
0<\vartheta_0<\vartheta_1<\cdots <\vartheta_n<\pi
\end{equation}
of $p_{n+1}(\cos \vartheta  ; p_+,q_+p_-, q_-| 0 )$  solve the Bethe Ansatz equation in Eq. \eqref{BAE-r} (with $\xi_j$ replaced by $\vartheta$). More specifically, the root $\vartheta_\texttt{k}$ corresponds to the solution of the associated  logarithmic Bethe Ansatz equation  (cf. Remark \ref{gradient:rem} above)
\begin{equation}\label{critical-eq-r}
2n\vartheta +v_{p_+}(\vartheta) +v_{q_+}(\vartheta)+v_{p_-}(\vartheta)+v_{q_-}(\vartheta)
=2\pi ( \texttt{k}+1 )  ,
\end{equation}
with $\texttt{k}\in\{ 0,1,\ldots ,n\}$.

The upshot is that to any $\kappa=(\kappa_1,\ldots,\kappa_m)\in\Lambda^{(m,n)}$,  we can now attach a solution $\tilde\xi_\kappa$  of   the  decoupled system of Bethe Ansatz equations in Eq. \eqref{BAE-r}  by forming the following  vector of $\rm{q}=0$ Askey-Wilson roots:
\begin{equation}\label{BA-solr}
\tilde\xi_\kappa=(\vartheta_{\kappa_1},\vartheta_{\kappa_2},\ldots,\vartheta_{\kappa_m})\in \{\xi\in\mathbb{R}^m\mid \pi>\xi_1\geq\xi_2\geq\cdots\geq\xi_m>0\}  .
\end{equation}
Notice that $\tilde\xi_\kappa$ encodes the unique global minimum
of the decoupled Morse function
{\small
\begin{align}\label{q=1:Morse}
&\tilde{V}^{}_\kappa(\xi_1,\ldots,\xi_m)=\\
&\sum_{1\leq j\leq m}  \left(
n\xi_j^2-2\pi \bigl( \kappa_j+1\bigr) \xi_j 
+   \int^{\xi_j}_0 \bigl(v_{p_+}(x)+v_{q_+}(x)+v_{p_-}(x)+v_{q_-}(x)\bigr)\text{d}x\right)  .\nonumber
\end{align}
}

\subsection{Spectrum and eigenbasis for $q\to 1$}

When tying the above observations together, one is led to the following $q\to 1$ counterpart of Theorem \ref{diagonalization:thm}.

\begin{theorem}[Spectrum and Eigenbasis for $q\to 1$]\label{diagonalization-r:thm}
Let the boundary parameters $\alpha_\pm$, $\beta_\pm$ be of the form in Eq. \eqref{bp:a} with $p_\pm,q_\pm\in (-1,1)$. For any $\xi\in\{ \xi\in \mathbb{R}^m\mid \xi_j\not\in\pi\mathbb{Z},\, \forall 1\leq j\leq m\}$ and
$\kappa\in\Lambda^{(m,n)}$,  the function $\tilde{\psi}_\xi :\Lambda^{(n,m)}\to \mathbb{R}$  refers to the $q=1$ Hall-Littlewood Bethe Ansatz wave function in Eq. \eqref{BAr}--\eqref{HL1} and
 $\tilde\xi_\kappa$ denotes the solution in Eq. \eqref{BA-solr} of the decoupled Bethe Ansatz equations  \eqref{BAE-r}.

(i) The spectrum of $\tilde{H}$ \eqref{H-r} in the Hilbert space $\ell (\Lambda^{(n,m)},\tilde{\Delta})$ consists of the eigenvalues $\tilde{E}(\tilde\xi_\kappa)$, $\kappa\in\Lambda^{(m,n)}$, where
$\tilde{E}(\xi)$ is given by Eq. \eqref{ev-eq-r}.

(ii) The corresponding Bethe Ansatz wave functions $\tilde\psi_{\tilde\xi_{\kappa}}$, $\kappa\in\Lambda^{(m,n)}$ constitute an orthogonal eigenbasis for
 $\tilde{H}$ \eqref{H-r}  in $\ell (\Lambda^{(n,m)},\tilde\Delta)$ such that
\begin{equation}\label{diagonal-r}
\tilde{H}\tilde{\psi}_{\tilde\xi_\kappa}= \tilde{E}(\tilde\xi_\kappa) \tilde{\psi}_{\tilde\xi_\kappa}\qquad (\kappa\in\Lambda^{(m,n)}).
\end{equation}
\end{theorem}

\begin{remark}
Systematic studies of orthogonal polynomials associated with the multinomial distribution give rise to multivariate generalizations of the 
Krawtchouk polynomials \cite{dia-gri:introduction,gen-vin-zhe:multivariate,gru-rah:system,ili:lie,ili-xu:hahn,miz-tan:hypergeometric}.  For the particular instance of the
two-parameter multinomial distribution in Eqs. \eqref{weights-r}--\eqref{norm-r}, Theorem \ref{diagonalization-r:thm} suggests an intriguing link to the $q=1$ Hall-Littlewood Bethe Ansatz wave function $\tilde\psi_{\tilde\xi_{\kappa}}$. When $n=1$ this link is actually well-understood in the literature as a relation between classical univariate Krawtchouk polynomials and elementary symmetric polynomials,
cf. e.g. \cite[Equation (5.14)]{die:q-deformation}.
\end{remark}

\subsection{Proof of Theorem \ref{diagonalization-r:thm}}
In order to establish the claims of the theorem in full rigor avoiding tricky formal limits, let  us first  check that the  Bethe Ansatz wave functions $\tilde\psi_{\tilde\xi_{\kappa}}$, $\kappa\in\Lambda^{(m,n)}$ are indeed orthogonal in the Hilbert space 
 $\ell (\Lambda^{(n,m)},\tilde\Delta)$ (i.e. with the weights $\tilde{\Delta}_\mu$ replacing $\Delta_\mu$):
%\begin{subequations}
\begin{equation}\label{orthogonalityr}
\forall \kappa,\nu\in\Lambda^{(m,n)}:\qquad \langle \tilde{\psi}_{\tilde{\xi}_\kappa},\tilde{\psi}_{\tilde{\xi}_\nu}\rangle_{\tilde{\Delta} }= 0\quad \text{if}\ \kappa\neq \nu .
\end{equation}
Rewritten in terms of  the $q=1$ hyperoctahedral Hall-Littlewood polynomials, the inner product on the LHS of Eq. \eqref{orthogonalityr}
becomes (cf. Eqs. \eqref{R-orthogonality}, \eqref{weights-conjugate}):
\begin{equation}\label{sum-nm}
 \sum_{\lambda \in\Lambda^{(m,n)} } 
 \frac{ \tilde{R}_{\lambda} (\tilde{\xi}_\kappa) \tilde{R}_{\lambda} (\tilde{\xi}_\nu )\, m! }{(1-\alpha_+)^{\mathrm{m}_0(\lambda)} 
 (1-\alpha_-){}^{\mathrm{m}_n(\lambda)}  
\prod_{0\leq i\leq n}  \mathrm{m}_i (\lambda)! } ,
 \end{equation}
 with $ \tilde{R}_{\lambda} (\xi) $ and $ \tilde{\xi}_\kappa $ taken from Eqs. \eqref{HLr} and \eqref{BA-solr}, respectively.
% \end{subequations}
In particular, if $m=1$ then the sum in Eq. \eqref{sum-nm} is of the form
\begin{equation}\label{sumn1}
 \sum_{0\leq \lambda \leq n} 
 \frac{R_{\lambda} (\vartheta_{\kappa}) R_{\lambda} (\vartheta_\nu) }{(1-\alpha_+)^{\delta_{\lambda} } (1-\alpha_-){}^{\delta_{n-\lambda} }}  
 \end{equation}
with
$0\leq \kappa\neq\nu \leq n$, where $\vartheta_0,\ldots,\vartheta_n$ refer to the $\textrm{q}=0$ Askey-Wilson  roots from Eq. \eqref{roots}.
Since the  sum in Eq. \eqref{sumn1} coincides with that of the inner product in Case (a) of Remark \ref{orthogonality:rem} (specialized to $m=1$), in this simplest situation
 the asserted orthogonality
is  immediate from the remark in question.

If on the other hand $m>1$, then the inner product in Eq. \eqref{sum-nm}  decomposes into  a sum of contributions of the form
\begin{align*}
&\sum_{\lambda \in\Lambda^{(m,n)} } 
 \frac{R_{\lambda_1}( \vartheta_{\texttt{k}_1}) \cdots  R_{\lambda_m}( \vartheta_{\texttt{k}_m})  
 R_{\lambda_1}( \vartheta_{\texttt{n}_1}) \cdots  R_{\lambda_m}( \vartheta_{\texttt{n}_m}) m!  }{(1-\alpha_+)^{\mathrm{m}_0(\lambda)} 
 (1-\alpha_-){}^{\mathrm{m}_n(\lambda)}  
\prod_{0\leq i\leq n}  \mathrm{m}_i (\lambda)! } \\
&= 
\sum_{0\leq \lambda_1,\ldots,\lambda_m\leq n}
\frac{R_{\lambda_1}( \vartheta_{\texttt{k}_1}) \cdots  R_{\lambda_m}( \vartheta_{\texttt{k}_m})  
 R_{\lambda_1}( \vartheta_{\texttt{n}_1}) \cdots  R_{\lambda_m}( \vartheta_{\texttt{n}_m})}
 {  \prod_{1\leq j \leq m}  (1-\alpha_+)^{\delta_{\lambda_j} } (1-\alpha_-  )^{\delta_{n-\lambda_j}}} \\
 &= \prod_{1\leq j\leq m} \sum_{0\leq \lambda_j\leq n} 
  \frac{R_{\lambda_j} (\vartheta_{\texttt{k}_j}) R_{\lambda_j} (\vartheta_{\texttt{n}_j}) }{(1-\alpha_+)^{\delta_{\lambda_j} } (1-\alpha_-){}^{\delta_{n-\lambda_j} }}  ,
 \end{align*}
where the $m$-tuple $(\texttt{k}_1, \texttt{k}_2,\ldots ,\texttt{k}_m)$ denotes a reordering $(\kappa_{\sigma(1)},\kappa_{\sigma(2)},\ldots,\kappa_{\sigma(m)})$ of  $\kappa\in\Lambda^{(m,n)}$
and  the $m$-tuple $(\texttt{n}_1,\texttt{n}_2,\ldots ,\texttt{n}_m)$ denotes a reordering of $\nu\in\Lambda^{(m,n)}$ (not necesarilly stemming from the same permutation $\sigma\in S_m$). From the orthogonality for $m=1$ it is now clear that all such contributions vanish, unless $\texttt{k}_j=\texttt{n}_j$ for $j=1,\ldots, m$, i.e., except when $\kappa=\nu$ and both  reorderings coincide.

It remains to verify that the eigenvalue equation in Eq. \eqref{ev-eq-r} is satisfied at $\xi=\tilde{\xi}_\kappa$ for $\kappa\in\Lambda^{(m,n)}$. 
Rewritten in terms of $\lambda=\mu^\prime\in\Lambda^{(m,n)}$, this eigenvalue equation reads explicitly (cf. Eq. \eqref{rec-rel}):
\begin{align}\label{rec-rel-r}
\tilde{E}(\tilde{\xi}_\kappa) \tilde{R}_\lambda (\tilde{\xi}_\kappa)
= & \bigl(\beta_+ \mathrm{m}_0(\lambda) +
\beta_-  \mathrm{m}_n(\lambda)   \bigr)
 \tilde{R}_\lambda (\tilde{\xi}_\kappa)\\
+&\sum_{\substack{1\leq j \leq m\\ \lambda+e_j\in\Lambda^{(m,n)}}}
(1-\alpha_+ )^{\delta_{\lambda_j} }
 \mathrm{m}_{\lambda_j}(\lambda) \tilde{R}_{\lambda+e_j} (\tilde{\xi}_\kappa) \nonumber \\
+&\sum_{\substack{1\leq j \leq m\\ \lambda-e_j\in\Lambda^{(m,n)}}}
(1-\alpha_- )^{\delta_{n-\lambda_j} }
m_{\lambda_j}(\lambda)  \tilde{R}_{\lambda-e_j}  (\tilde{\xi}_\kappa) . \nonumber
\end{align}
For instance, if $m=1$ then Eq. \eqref{rec-rel-r} simplifies to
\begin{align}\label{rec-rel-r-m=1}
2\cos& (\vartheta_\kappa)  R_\lambda (\vartheta_\kappa)
=  \bigl(\beta_+ \delta_\lambda +
\beta_-  \delta_{n-\lambda}  \bigr)
  R_\lambda (\vartheta_\kappa)+ \\
&
(1-\alpha_+ )^{\delta_{\lambda} }
 (1-\delta_{n-\lambda}) R_{\lambda+1} ( \vartheta_\kappa) 
+
(1-\alpha_- )^{\delta_{n-\lambda} }
(1-\delta_\lambda) R_{\lambda-1}  (\vartheta_\kappa) ,  \nonumber
\end{align}
where $0\leq\kappa, \lambda\leq n$. Apart from a missing overall factor $(1-q)$ on both sides (which was actually divided out  at the start of Subsection \ref{sec5.1}), Eq. \eqref{rec-rel-r-m=1} coincides precisely with the
$m=1$ specialization of the eigenvalue equation from Theorem \ref{diagonalization:thm} at $\mu=\lambda^\prime$ (which agrees with the observation that the dependence on $q$ drops out when $m=1$).  Upon recalling Remark \ref{smooth:rem},
this settles the validity of Eq. \eqref{rec-rel-r-m=1} for the full parameter regime $p_\pm,q_\pm\in (-1,1)$.

Moreover, by virtue of   Eq.  \eqref{rec-rel-r-m=1} one has more generally that for $m>1$ and any $\sigma\in S_m$:
\begin{align*}
&2\sum_{1\leq j\leq m} \cos (\vartheta_{\kappa_j} )     \prod_{1\leq l\leq m} R_{\lambda_l} (\vartheta_{\kappa_{\sigma(l)}})
= \\
& \sum_{1\leq j\leq m} \ \bigl(\beta_+ \delta_{\lambda_j}+
\beta_-  \delta_{n-\lambda_j}   \bigr)
 \prod_{1\leq l\leq m} R_{\lambda_l} (\vartheta_{\kappa_{\sigma(l)}})   \\
+&\sum_{1\leq j \leq m}
(1-\alpha_+ )^{\delta_{\lambda_j} } (1- \delta_{n-\lambda_j} )
R_{\lambda_j+1} (\vartheta_{\kappa_{\sigma(j)}}) \prod_{\substack{1\leq l\leq m\\l\neq j}} R_{\lambda_l} (\vartheta_{\kappa_{\sigma(l)}})
  \nonumber \\
+&\sum_{1\leq j \leq m}
(1-\alpha_- )^{\delta_{n-\lambda_j} }
(1- \delta_{\lambda_j} )  R_{\lambda_j-1} (\vartheta_{\kappa_{\sigma(j)}})  \prod_{\substack{1\leq l\leq m\\l\neq j}} R_{\lambda_l} (\vartheta_{\kappa_{\sigma(l)}})  \\
=&  \bigl(\beta_+ \mathrm{m}_0(\lambda) +
\beta_-  \mathrm{m}_n(\lambda)   \bigr)
 \prod_{1\leq l\leq m} R_{\lambda_l} (\vartheta_{\kappa_{\sigma(l)}})    \\
+&\sum_{\substack{1\leq j \leq m\\ \lambda+e_j\in\Lambda^{(m,n)}}}
(1-\alpha_+ )^{\delta_{\lambda_j} }
 \mathrm{m}_{\lambda_j}(\lambda) R_{\lambda_j+1}  (\vartheta_{\kappa_{\sigma(j)}}) \prod_{\substack{1\leq l\leq m\\l\neq j}} R_{\lambda_l} (\vartheta_{\kappa_{\sigma(l)}})  \nonumber \\
+&\sum_{\substack{1\leq j \leq m\\ \lambda-e_j\in\Lambda^{(m,n)}}}
(1-\alpha_- )^{\delta_{n-\lambda_j} }
m_{\lambda_j}(\lambda) R_{\lambda_j-1}  (\vartheta_{\kappa_{\sigma(j)}}) \prod_{\substack{1\leq l\leq m\\l\neq j}} R_{\lambda_l} (\vartheta_{\kappa_{\sigma(l)}}),
\end{align*}
which entails Eq. \eqref{rec-rel-r} through symmetrization by summing over all $\sigma\in S_m$ on both sides.

\section*{Acknowledgements}
This work was supported in part by the {\em Fondo Nacional de Desarrollo
Cient\'{\i}fico y Tecnol\'ogico (FONDECYT)} Grant \# 1210015.

\bibliographystyle{amsplain}

\end{document}